\documentclass[conference]{IEEEtran}

\usepackage{amsfonts}
\usepackage{amssymb}
\usepackage{stfloats}
\usepackage{graphicx}
\usepackage{amsmath}
\usepackage{array}
\usepackage{epstopdf}
\usepackage{amsthm}
\usepackage{algorithm}
\usepackage[noend]{algpseudocode}
\usepackage{amsmath}
\usepackage{graphics}
\usepackage{epsfig}

\usepackage{bm}
\usepackage{upgreek}
\usepackage{caption}
\usepackage{subcaption}
\usepackage{tabularx}
\usepackage{cite}
\usepackage{lipsum,multicol}

\usepackage[usenames]{color}

\newcommand{\hermconj}{^{\mathsf{H}}}

\newcommand{\re}{\operatorname{Re}}

\newcommand{\st}{\operatorname{s.t.}}
\newcommand{\diag}{\operatorname{diag}}

\DeclareMathOperator{\find}{find}
\DeclareMathOperator{\argmin}{argmin}

\usepackage{mathtools}

\DeclarePairedDelimiter\floor{\lfloor}{\rfloor}

\newtheorem{Lemma}{Lemma}

\newtheorem{Remark}{Remark}

\usepackage{geometry}
\IEEEoverridecommandlockouts
\begin{document}

\title{Reconfigurable Intelligent Surface Assisted \\ Edge Machine Learning}

\author{
    \IEEEauthorblockN {Shanfeng Huang\IEEEauthorrefmark{1}\IEEEauthorrefmark{2}, Shuai Wang\IEEEauthorrefmark{1}, Rui Wang\IEEEauthorrefmark{1}, Miaowen Wen\IEEEauthorrefmark{3} and Kaibin Huang\IEEEauthorrefmark{2}}

    \IEEEauthorblockA{
        \IEEEauthorrefmark{1}Department of Electrical and Electronic Engineering, Southern University of Science and Technology\\
        \IEEEauthorrefmark{2}Department of Electrical and Electronic Engineering, The University of Hong Kong\\
        \IEEEauthorrefmark{3}School of Electronic and Information Engineering, South China University of Technology\\
	    Email: \{sfhuang, huangkb\}@eee.hku.hk, \{wangs3,wang.r\}@sustech.edu.cn, eemwwen@scut.edu.cn}

    \thanks{	
        This work was supported in part by the National Natural Science Foundation of China under Grant 62001203, in part by the Shenzhen Fundamental Research Program under Grant JCYJ20190809142403596, and in part by the Fundamental Research Funds for the Central Universities under Grant 2019SJ02.}
}

\maketitle

\begin{abstract}
The ever-growing popularity and rapid improving of artificial intelligence (AI) have raised rethinking on the evolution of wireless networks.  Mobile edge computing (MEC) provides a natural platform for AI applications since it provides rich computation resources to train AI models, as well as low-latency access to the data generated by mobile and Internet of Things devices. In this paper, we present an infrastructure to perform machine learning tasks at an MEC server with the assistance of a reconfigurable intelligent surface (RIS). In contrast to conventional communication systems where the principal criteria are to maximize the throughput, we aim at optimizing the learning performance. Specifically, we minimize the maximum learning error of all  users by jointly optimizing the beamforming vectors of the base station and the phase-shift matrix of the RIS. An alternating optimization-based framework is proposed to optimize the two terms iteratively, where closed-form expressions of the beamforming vectors are derived, and an alternating direction method of multipliers (ADMM)-based algorithm is designed together with an error level searching framework to effectively solve the nonconvex optimization problem of the phase-shift matrix. Simulation results demonstrate significant gains of deploying an RIS and validate the advantages of our proposed algorithms over various benchmarks.
\end{abstract}


\IEEEpeerreviewmaketitle

\section{Introduction}
The prevalence of mobile terminals and rapid growth of Internet of Things (IoT) technology have boosted a wide spectrum of new applications, many of which are computation-intensive and latency-critical, such as image recognition, mobile augmented reality, and edge machine intelligence. 
Mobile edge computing (MEC) is naturally well-suited for the AI-oriented networks, and  the marriage of mobile edge computing (MEC) and AI has given rise to a new research area, called  ``edge intelligence (EI)'' or ``edge AI'' \cite{Zhou2019EI,Li2019EdgeAI,Zhu2020ELsurvey,Yu2020IE}. 
Moreover, to overcome wireless channel hostilities, an emerging paradigm called reconfigurable intelligent surface (RIS) was proposed, aiming at creating a smart radio environment by turning the wireless environment into an optimization variable, which can be controlled and programmed\cite{Renzo2019SRE}. 
Hence, we would like to investigate the design of an RIS-assisted edge learning system. 

In contrast with conventional communication systems where the  general goals are to maximize the throughput, edge ML systems aim at optimizing the learning performance. As a result, the well-known resource allocation schemes that are optimized for conventional systems, such as water-filling and max-min fairness schemes may lead to poor learning performance since they do not take into account the learning-specific factors such as model and data complexities. Recently, there are some outstanding works that aim at optimizing the resource allocation schemes for learning-centric systems. In \cite{Liu2020dataimportance}, the authors proposed a data-importance aware user scheduling scheme for edge ML systems, where data are regarded as having different importance levels based on certain importance measurement. Nevertheless, the analysis is mainly based on SVM. For more general ML models, the importance of training data is hard to quantify. In \cite{Shi2019RISEL}, the authors investigated an RIS-assisted edge inference system. However, the inference tasks are considered as general edge computing tasks in essence, leading to few insights for real ML tasks. More recently, our previous work \cite{Wang2020EdgeLearningTWC}  put forth and validated a nonlinear classification error model for ML tasks, based on which a learning-centric power allocation scheme was proposed and shown to outperform conventional resource allocation schemes significantly with respect to learning error. In this paper, we further extend \cite{Wang2020EdgeLearningTWC} to the scenario where an RIS is deployed to provide intelligence to the wireless channels. With the presence of the RIS, new challenges in the beamforming vector and phase shift optimization arise.

In this paper, we shed light on the design of RIS-assisted edge ML with heterogeneous learning tasks. Specifically, we adopt the nonlinear learning error model in \cite{Wang2020EdgeLearningTWC,johnson2018accuracypilotdata}, and aim at minimizing the maximum learning error of all the learning tasks by jointly optimizing the beamforming vectors at the base station (BS) and the phase shift matrix at the RIS. The optimization problem is nonconvex and involves many optimization variables. To address this challenge, we design an alternating optimization (AO)-based framework to decompose the primal problem and each subproblem is efficiently solved either in closed form or with low-complexity algorithms. 
Specifically, the optimization of beamforming vectors is shown to be equivalent to maximizing the signal-to-interference-plus-noise ratios (SINRs), and closed-form expressions are derived. 
To solve the phase-shift matrix optimization problem, we propose an error level searching (ELS)-based framework to transform the exponential objective into SINR constraints, and exploit alternating direction method of multipliers (ADMM) to decouple the problem to a set of subproblems that can be solved in a distributed manner. Simulations on well-known ML models and public datasets verify the nonlinear learning error model, and demonstrate that our proposed scheme can achieve significantly lower learning error than that of various benchmarks. 



\section{System Model}
We consider an edge ML system as shown in Fig. \ref{fig:SystemModel}, where an intelligent edge server attached to a BS with $N$ antennas is serving $K$ single-antenna users, each with an ML task.
The communication is assisted by an RIS, consisting of $M$ passive reflecting elements which could rotate the phase of the incident signal waves.
In particular, the edge server is designated to train $K$ classification models by collecting data observed at the $K$ mobile users. The classification models can be CNNs, SVMs, etc. 
\begin{figure}[tb]
    \centering
    \includegraphics[scale=0.3]{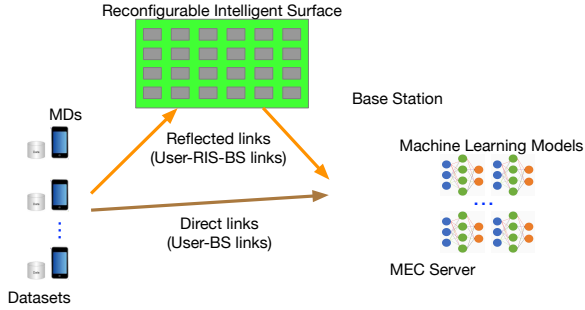}
    \caption{An RIS-assisted edge ML system.}
    \label{fig:SystemModel}
\end{figure}

The training data are transmitted from the mobile users to the edge server via wireless channels which have intrinsic random feature due to multi-path effect and  can suffer from high propagation loss \cite{goldsmith_2005}. To this end, this paper considers an RIS-assisted scheme that can configure the channel intelligently by tuning the phase shifts of the reflecting elements adaptively. With the presence of the RIS, the channel from user $k$ to the BS includes both the direct link (user-BS link) and the reflected link (user-RIS-BS link), where the reflected link consists of the user-RIS link, the phase shifts at RIS, and the RIS-BS link \cite{QWu2019IRSBeamforming}. Denote the channel vector from $k$-th user to the BS as $\mathbf h_k$. It can be expressed as
\begin{align}
    \mathbf h_k=\underbrace{\mathbf h_{\text{d}, k}}_\text{direct link}+\underbrace{\mathbf G\hermconj \mathbf{\Theta}\hermconj \mathbf h_{\text{r},k}}_\text{reflected link},
\end{align}
where $\mathbf h_{\text{d}, k}\in \mathbb C^{N\times 1}$, $\mathbf h_{\text{r},k}\in \mathbb C^{M\times 1}$, and $\mathbf G\in \mathbb C^{M\times N}$ denote the channel vectors and matrix from user $k$ to the BS, from user $k$ to the RIS, and from the RIS to the BS, respectively. Moreover, $\mathbf{\Theta}=\beta \text{diag}(e^{j\varphi_1},\cdots, e^{j\varphi_M})\in \mathbb C^{M\times M}$ denotes the phase-shift matrix of the RIS, where $\beta\in [0, 1]$ is the amplitude reflection coefficient and $\varphi_{m}\in [0,2\pi)$ is the phase shift of the $m$-th reflecting element. Without loss of generality, $\beta$ is typically set to 1. 

Denote the transmitted signal of  user $k \in \{1, 2, \cdots, K\}$ as $x_k$ with power $\mathbb E[|x_k|^2]=p_k$. Accordingly, the received signal $\mathbf y=[y_1, \cdots, y_N]\in \mathbb C^{N\times 1}$ at the BS can be written as 
\begin{align}
    \mathbf y=\sum_{k=1}^K \mathbf h_k x_k+\mathbf n,
\end{align}
where $\mathbf n\sim \mathcal{CN}(\mathbf 0, \sigma^2 \mathbf I_N)$ is the additive white Gaussian noise (AWGN) at the BS. 
A beamforming vector $\mathbf w_k$ with $\mathbf w_k\hermconj \mathbf w_k=1$ is applied for the received signal from each user $k$. Thus, the estimated symbol at the BS for user $k$ is given by
\begin{align}
    \hat y_k=\mathbf w_k\hermconj \mathbf y=\mathbf w_k\hermconj \mathbf h_k x_k+ \sum_{i=1, i\neq k}^K \mathbf w_k\hermconj \mathbf h_{i}x_{i} +\mathbf w_k\hermconj \mathbf n.  
\end{align}
Accordingly, the achievable spectral efficiency of user $k$ in terms of bps/Hz is given by 
\begin{align}
    R_k=\log_2\left(1+\frac{p_k |\mathbf w_k\hermconj \mathbf h_k|^2}{\sum_{i=1, i\neq k}^{K}p_{i}|\mathbf w_k\hermconj \mathbf h_{i}|^2+\sigma^2}\right).
\end{align}

Let $B$ denote the bandwidth of the considered system, and $T$ the total transmission time. Thus, the total number of data samples for user $k$'s task is given by 
\begin{align}
    v_k=\floor*{\frac{BTR_k}{D_k}}\approx  \frac{BTR_k}{D_k},
\end{align}
where $D_k$ is the number of bits for each data sample, and the approximation is due to $\floor*{x}\to x$ when $x\gg 1$.


\section{Problem Formulation}
In contrast with the conventional communication systems where the principal design criterion is usually to maximize the throughput, edge ML systems aim at maximizing the learning performance. Specifically, in the edge ML system considered herein, we aim at minimizing the maximum learning error of all the participating users by jointly optimizing the beamforming vectors $\{\mathbf w_k\}_{k=1}^K$ at the BS, and the phase-shift matrix $\mathbf \Theta$ of the RIS. Thus, we have the following optimization problem.
\begin{subequations}
\begin{align}
    \mathcal P: \min_{\{\mathbf w_k\}_{k=1}^K,\mathbf{\Theta}, \mathbf v} \quad &\max_{k=1,\cdots,K}\quad\Psi_k(v_k) \nonumber\\
    \st \quad\quad 
    &\mathbf w_k\hermconj \mathbf w_k=1, \quad k=1,\cdots,K, \label{eq:bf}\\
    &\frac{BTR_k}{D_k}=v_k, \quad k=1, \cdots, K, \label{eq:v_m}\\
    &0\leq \varphi_m<2\pi,\quad m=1, \cdots, M,
\end{align}
\end{subequations}
where $\Psi_k(v_k)$ is the classification error of learning model $k$ given the sample size $v_k$. In general, the functions $\{\Psi_1, \cdots, \Psi_K\}$ can hardly be expressed analytically. Propitiously, their approximate expressions can be obtained based on the analysis in \cite{Wang2020EdgeLearningTWC, johnson2018accuracypilotdata,BELEITES2013samplesize}. Here, we simply adopt the non-linear model developed in \cite{Wang2020EdgeLearningTWC}, i.e., 
\begin{align}\label{eq:errormodel}
    \Psi_k(v_k)\approx c_k v_k^{-d_k},
\end{align}
where $c_k$ and $d_k$ are tuning parameters which can be obtained by curve fitting.

By substituting (\ref{eq:v_m}) and (\ref{eq:errormodel}) into the objective function, problem $\mathcal P$ is transformed into the following problem.
\begin{subequations}
    \begin{align}
        \mathcal P1: &\min_{\{\mathbf w_k\}_{k=1}^K,\mathbf{\Theta}}\max_{k=1,\cdots,K} c_k\bigg[\frac{BT}{D_k}\log_2\bigg(1 \nonumber\\
        &+\frac{|\mathbf w_k\hermconj(\mathbf h_{\text{d},k}+\mathbf G\hermconj \mathbf{\Theta}\hermconj \mathbf h_{\text{r},k})|^2p_k}{\sum_{i=1, i\neq k}^K |\mathbf w_k\hermconj (\mathbf h_{\text{d},i}+\mathbf G\hermconj \mathbf{\Theta}\hermconj \mathbf h_{\text{r},i})|^2 p_{i}+\sigma^2}\bigg)\bigg]^{-d_k}\nonumber\\
        \st \quad 
        &\mathbf w_k\hermconj \mathbf w_k=1, \quad k=1,\cdots,K, \label{eq:beam}\\
        &|\theta_m|=1,\quad m=1,\cdots,M.
    \end{align}
\end{subequations}

\begin{Remark}[Scaling law with large number of reflecting elements]
To gain some insights on how the number of reflecting elements affect the learning accuracy, we consider the case with a single user and a single-antenna BS, i.e., $K=1$ and $N=1$, and ignore the direct link. Thus, $\mathbf G$ becomes a vector and is denoted by $\mathbf g$.  The receive SNR becomes $p|\mathbf h_{\text{r}}\hermconj \mathbf\Theta\mathbf g|/{\sigma}^2$. Assume $\mathbf\Theta=\mathbf I_M$, $\mathbf h_{\text{r}}\sim\mathcal{CN}(\mathbf 0, \varrho_h^2\mathbf I_M)$, and $\mathbf g\sim\mathcal{CN}(\mathbf 0, \varrho_g^2\mathbf I_M)$. According to the central limit theorem, we have $\mathbf h_{\text{r}}\hermconj\mathbf g\sim\mathcal{CN}(\mathbf 0, M\varrho_h^2\varrho_g^2)$ as $M\to\infty$. Thus, the average receive SNR is $\mathbb E_{\mathbf h_{\text{r}},\mathbf g}p\mathbb|\mathbf h_{\text{r}}\hermconj \mathbf\Theta\mathbf g|/{\sigma}^2= Mp\varrho_h^2\varrho_g^2$. This indicates that the learning error is asymptotically proportional to $(\log_2(M))^{-d}$.
\end{Remark}

\section{Joint Beamforming and Phase-Shifter Design}
Note that problem $\mathcal P1$ is highly nonconvex due to the nonlinear learning error model in the objective function and the unit-modulus constraints. Moreover, the large number of optimization variables make the problem even more untractable. Fortunately, the optimization of the beamforming vectors and the phase-shift matrix can be decomposed. Hence, we adopt an AO-based algorithm to solve $\mathcal P1$ in an iterative manner via alternatively optimizing $\{\mathbf w_k\}_{k=1}^K$ and $\mathbf\Theta$.

\subsection{Beamforming Vectors Optimization}

Note that given $\mathbf\Theta$, the objective function of the original problem $\mathcal P1$ is still nonconvex in $\mathbf w_k$. However, since the objective function is monotonically decreasing in the SINR of each user and is decomposable with respect to $k$, the optimization of $\mathbf w_k$ with fixed $\mathbf \Theta$ can be equivalently solved by maximizing the SINR of each user $k$. Consequently, the optimal beamforming vectors can be obtained by solving the following $K$ subproblems.
\begin{subequations}
    \begin{align}
        \mathcal P_{\mathbf w_k}: \max_{\mathbf w_k} \quad &\frac{|\mathbf w_k\hermconj(\mathbf h_{\text{d},k}+\mathbf G\hermconj \mathbf{\Theta}\hermconj \mathbf h_{\text{r},k})|^2p_k}{\sum_{i=1, i\neq k}^K |\mathbf w_k\hermconj (\mathbf h_{\text{d},i}+\mathbf G\hermconj \mathbf{\Theta}\hermconj \mathbf h_{\text{r},i})|^2 p_{i}+\sigma^2}\nonumber\\
        \st\quad &\mathbf w_k\hermconj \mathbf w_k=1.
    \end{align}
\end{subequations}

Although each problem $\mathcal P_{\mathbf w_k}$ is still nonconvex in $\mathbf w_k$, its optimal solution can be achieved in closed-form as given in the following lemma. 

\begin{Lemma}\label{lem:opt_bf}
    Given $\mathbf\Theta$, the optimal solution of $\mathcal P_{\mathbf w_k}$ for arbitrary $k$ is given in closed-form by 
    \begin{align}\label{eq:optbeam}
        \mathbf w_k^{\diamond}=\frac{\left(\mathbf I_N+\sum_{i=1}^K\frac{p_{i}}{\sigma^2}\mathbf h_{i}\mathbf h_{i}\hermconj\right)^{-1}\mathbf h_k}{\left\|\left(\mathbf I_N+\sum_{i=1}^K\frac{p_{i}}{\sigma^2}\mathbf h_{i}\mathbf h_{i}\hermconj\right)^{-1}\mathbf h_k\right\|_2},
    \end{align}  
    where $\mathbf h_i=\mathbf h_{\text{d},i}+\mathbf G\hermconj \mathbf{\Theta}\hermconj \mathbf h_{\text{r},i}$, for $i=1,\cdots,K$.
\end{Lemma}
\begin{proof}
    The proof is similar to that in \cite{Bjornson2014MU-BF} and is neglected here due to page limitation.
\end{proof}

\subsection{Phase-shift Matrix Optimization}
Given the beamforming vectors $\{\mathbf w_k\}_{k=1}^K$, there remain only the unit-modulus constraints of the RIS elements. By exploiting $\mathbf\Theta=\diag(\bm\theta)$ and setting $\mathbf a_{k,i}=\beta \diag(\mathbf h_{\text{r},i}\hermconj)\mathbf G\mathbf w_k$, $b_{k,i}=\mathbf h_{\text{d},i}\hermconj \mathbf w_k$, the optimization of phase-shift matrix $\mathbf\Theta$ can be equivalently written as the following problem.
\begin{subequations}
    \begin{align}
        \!\!\!\mathcal P_{\bm\theta}:\!  &\min_{\bm\theta}\max_{k=1,\cdots,K} \!c_k\!\!\left[\!\!\frac{BT}{D_k}\!\log_2
        \!\!\left(\!\!1\!\!+\!\!\frac{|\bm\theta\hermconj\mathbf a_{k,k}\!+\!b_{k,k}|^2p_k}{\sum\limits_{i=1,i\neq k}^K|\bm\theta\hermconj\mathbf a_{k,i}\!+\!b_{k,i}|^2p_i\!+\!\sigma^2}\!\!\right)\!\!\right]^{-d_k}\nonumber\\
        &\st \quad |\theta_m|= 1, \forall m=1,\cdots,M.
    \end{align}   
\end{subequations}

A common approach to address the nonconvex unit-modulus constraints is semidefinite relaxation (SDR). Nevertheless, even SDR can circumvent the nonconvex unit-modulus constraints, the objective function remains nonconvex due to the nonlinear learning error model. Moreover, the solution achieved by SDR generally does not conform to the rank-1 constraint, and large number of Gaussian randomizations are required to find a rank-1 solution, which increases the complexity dramatically. Besides, SDR lifts the optimization variable from an $M\times 1$ vector to an $M\times M$ matrix.  Thus, SDR cannot scale up the number of RIS elements. To this end, we propose an ELS framework and an ADMM-based algorithm to solve problem $\mathcal P_{\bm\theta}$. Specifically, we first define the error level of the $k$-th ML task for all $k$ as 
\begin{align}
\delta_k\!=\!c_k\left[\frac{BT}{D_k}\log_2\!
\left(\!1\!+\!\frac{|\bm\theta\hermconj\mathbf a_{k,k}+b_{k,k}|^2p_k}{\sum_{i=1,i\neq k}^K|\bm\theta\hermconj\mathbf a_{k,i}\!+\!b_{k,i}|^2p_i\!+\!\sigma^2}\!\right)\!\right]^{-d_k}.
\end{align}
Thus, the maximum error level of all participating tasks is given by $\delta=\max_{k\in\mathcal K} \delta_k$. Then, for a given error level $\delta$, problem $\mathcal P_{\bm\theta}$ can be equivalently transformed to the following feasibility problem.
\begin{subequations}
    \begin{align}
        \mathcal P'_{\bm\theta}:  &\find\quad\bm\theta\\
        &\st\quad\frac{|\bm\theta\hermconj\mathbf a_{k,k}+b_{k,k}|^2p_k}{\sum_{i=1,i\neq k}^K|\bm\theta\hermconj\mathbf a_{k,i}+b_{k,i}|^2p_i+\sigma^2}\geq \gamma_k,\forall k\\
        &\quad\quad\quad |\theta_m|= 1,\quad \forall m,
    \end{align}   
\end{subequations}
where $\gamma_k=2^{\frac{D_k\left(\frac{c_k}{\delta}\right)^{\frac{1}{d_k}}}{BT}}-1$. If problem $\mathcal P_{\bm\theta}'$ is feasible, we can reduce $\delta$; otherwise, we increase $\delta$ to make $\mathcal P_{\bm\theta}'$ feasible, until $\delta$ converges to a certain value. We call this procedure error level searching (ELS).

In the sequel, we design an ADMM-based algorithm to solve problem $\mathcal P_{\bm\theta}'$. By introducing a series of auxiliary variables $\{\mathbf q_k\}_{k=1}^K$ and a new constraint $\mathbf q_1=\mathbf q_2=\cdots=\mathbf q_K=\bm\theta$, problem $\mathcal P_{\bm\theta}'$ can be further rewritten as the following form.
\begin{subequations}
    \begin{align}
        \find &\quad\{\mathbf q_k\}_{k=1}^K, \bm\theta\\
        \st &\frac{|\mathbf q_k\hermconj\mathbf a_{k,k}+b_{k,k}|^2p_k}{\sum_{i=1,i\neq k}^K|\mathbf q_k\hermconj\mathbf a_{k,i}+b_{k,i}|^2p_i+\sigma^2}\geq \gamma_k,\ k=1,\cdots,K\label{eq:snr_constraint}\\
        &|\theta_m|= 1, \quad m=1,\cdots,M\label{eq:norm_constraint}\\
        &\mathbf q_k=\bm\theta, \quad k=1,\cdots,K.
    \end{align}  
    \label{prob:theta_admm} 
\end{subequations}
The augmented Lagrangian (using the scaled dual variable) of problem (\ref{prob:theta_admm}) is given by 
\begin{align*}
    \mathcal L_{\rho}\!(\!\mathbf q_1\!,\!\cdots\!,\!\mathbf q_K\!,\!\bm\theta\!,\!\mathbf u_1\!,\!\cdots\!,\!\mathbf u_K\!)\!\!=\!\!\!\sum_{k=1}^K \!\mathbb I_{\mathcal B_k}\!(\!\mathbf q_k\!)\!\!+\!\! \mathbb I_{\mathcal C}\!(\!\bm\theta\!)\!\!
    +\!\!\rho\!\!\sum_{k=1}^K\!\|\mathbf q_k\!\!-\!\!\bm\theta\!+\!\mathbf u_k\!\|\!^2,
\end{align*}
where $\mathcal B_k$ is the feasibility region of the $k$-th constraint in (\ref{eq:snr_constraint}) and $\mathcal C$ is the feasibility region of constraint (\ref{eq:norm_constraint}), $\rho>0$ is the penalty parameter, and $\mathbf u_k$ is the scaled dual variable. 
Moreover, $\mathbb I$ is the indicator function with $\mathbb I_{\mathcal X}(\mathbf x)=0$ if $\mathbf x\in\mathcal X$ and $+\infty$ otherwise.

The ADMM algorithm iteratively update $\mathbf q_k$, $\bm\theta$ and $\mathbf u_k$ as follows, until a feasible solution is found.
\begin{subequations}\label{eq:admm_update}
    \begin{align}
        &\mathbf q_k^{t+1}:=\!\argmin_{\mathbf q_k} \!\!\mathcal L_{\rho}(\!\mathbf q_1,\!\cdots\!,\!\mathbf q_K,\!\bm\theta^t,\!\mathbf u_1^t,\cdots,\mathbf u_K^t), \forall k\\
        &\bm\theta^{t+1}:=\argmin_{\bm\theta} \mathcal L_{\rho}(\mathbf q_1^{t+1},\cdots,\mathbf q_K^{t+1},\bm\theta,\mathbf u_1^t,\cdots,\mathbf u_K^t)\\
        &\mathbf u_k^{t+1}:=\mathbf u_k^t+\mathbf q_k^{t+1}-\bm\theta^{t+1}, \forall k
    \end{align}
\end{subequations}

In the sequel, we show that each update in (\ref{eq:admm_update}) can be efficiently solved either in closed-form or with low complexity.

1) $\mathbf q_k$ update: The update of $\mathbf q_k$ can be equivalently written as the following problem after removing the irrelevant terms.
\begin{align}\label{eq:q-update}
    \mathbf q_k^{t+1}\!\!=\!\argmin_{\mathbf q_k}\sum_{k=1}^K \!\mathbb I_{\mathcal A_k}(\mathbf q_k)\!+\!\rho\sum_{k=1}^K\|\mathbf q_k\!-\!\bm\theta^t\!+\!\mathbf u_k^t\|^2.
\end{align}
Note that the update of $\mathbf q_k$ can be decoupled into $K$ subproblems for each $k$.
\begin{subequations}
\begin{align}
    \min_{\mathbf q_k} \quad&\|\mathbf q_k-\bm\theta^t+\mathbf u_k^t\|^2\\
    \st \quad&\frac{|\mathbf q_k\hermconj\mathbf a_{k,k}+b_{k,k}|^2p_k}{\sum_{i=1,i\neq k}^K|\mathbf q_k\hermconj\mathbf a_{k,i}+b_{k,i}|^2p_i+\sigma^2}\geq \gamma_k.
\end{align}
\label{prob:q_update}
\end{subequations}

Although problem (\ref{prob:q_update}) is nonconvex in general, strong duality holds and the Lagrangian relaxation produces the optimal solution since there is only one constraint \cite{boyd2004convex}. Thus, we can solve it efficiently using the Lagrangian dual method. Rephrasing problem (\ref{prob:q_update}), it can be equivalently written as the following compact form.
\begin{subequations}
    \begin{align}
        \min_{\mathbf q_k} \quad&\|\mathbf q_k-\bm\zeta_k^t\|^2\\
        \st \quad&\mathbf q_k\hermconj \mathbf A_k\mathbf q_k-2\re\{\mathbf b_k\hermconj\mathbf q_k\}= \tau_k,
    \end{align}
    \label{prob:q_qcqp1}
\end{subequations}
where $\bm\zeta_k^t=\bm\theta^t-\mathbf u_k^t$, $\mathbf A_k=\gamma_k\sum_{i=1,i\neq k}^K\mathbf a_{k,i}\mathbf a_{k,i}\hermconj p_i-\mathbf a_{k,k}\mathbf a_{k,k}\hermconj p_k$, $\mathbf b_k=\mathbf a_{k,k}b_{k,k}^* p_k-\gamma_k\sum_{i=1,i\neq k}^K\mathbf a_{k,i}b_{k,i}^*p_i$, and $\tau_k=|b_{k,k}|^2p_k-\gamma_k\sum_{i=1,i\neq k}^K|b_{k,i}|^2p_i-\gamma_k\sigma^2$. Note that we have changed the constraint to equality to simplify the follow-up derivations. When considering the inequality constraint, we can just check whether $\mathbf q_k=\bm\zeta_k^t$ is feasible. If yes, $\mathbf q_k^*=\bm\zeta_k^t$ is the optimal solution; if not, the optimal solution must satisfy the equality constraint.

For ease of notation, we neglect the subscript $k$ in problem (\ref{prob:q_qcqp1}), and let $\mathbf A=\mathbf Q\mathbf\Lambda\mathbf Q\hermconj$ be the eigenvalue decomposition. Then, problem (\ref{prob:q_qcqp1}) is equivalent to 
\begin{subequations}
    \begin{align}
        \min_{\tilde{\mathbf q}}\quad&\|\tilde{\mathbf q}-\tilde{\bm\zeta}^t\|^2\\
        \st \quad&\tilde{\mathbf q}\hermconj \mathbf\Lambda\tilde{\mathbf q}-2\re\{\tilde{\mathbf b}\hermconj\tilde{\mathbf q}\}= \tau,
    \end{align}
    \label{prob:q_qcqp1_equ}
\end{subequations}
where $\tilde{\mathbf q}=\mathbf Q\hermconj\mathbf q$, $\tilde{\bm\zeta}^t=\mathbf Q\hermconj\bm\zeta^t$, and $\tilde{\mathbf b}=\mathbf Q\hermconj\mathbf b$. 

As a result, the optimal solution can be efficiently found by the following lemma.
\begin{Lemma}\label{lem:opt_q}
    The optimal solution of problem (\ref{prob:q_qcqp1_equ}) is given by 
    \begin{align}\label{eq:q-opt}
        \tilde{\mathbf q}^*=(\mathbf I+\mu\mathbf\Lambda)^{-1}(\tilde{\bm\zeta}+\mu\tilde{\mathbf b}),
    \end{align}
    where $\mu$ is the Lagrangian multiplier of problem (\ref{prob:q_qcqp1_equ}). Moreover, $\mu$ can be found by solving a nonlinear equation $\chi(\mu)=0$ with 
    \begin{align*}
        \chi(\mu)\!=\!\sum_{m=1}^M\!\lambda_m\left|\frac{\tilde{\zeta}_m+\mu\tilde b_m}{1+\mu\lambda_m}\right|^2\!\!-\!2\re\left\{\sum_{m=1}^M\tilde b_m^*\frac{\tilde{\zeta}_m+\mu\tilde b_m}{1+\mu\lambda_m}\right\}\!-\!\tau,
    \end{align*}
    where $\lambda_m$ is the $m$-th diagonal entry of $\mathbf\Lambda$.
\end{Lemma}
\begin{proof}
    Please refer to Appendix A.
\end{proof}

Taking derivative on $\chi(\mu)$ with respect to $\mu$, we have 
\begin{align}
    \chi'(\mu)=-2\sum_{m=1}^M\frac{|\tilde b_m-\lambda_m\tilde\zeta_m|^2}{(1+\mu\lambda_m)^3}.
\end{align}
Since we assume the feasibility of problem (\ref{prob:q_qcqp1_equ}), there must exist $\mu$ with $\mathbf I+\mu\mathbf\Lambda\succeq 0$, such that value of $\tilde{\mathbf q}$ minimizing the Lagrangian also satisfies the equality constraint. Thus, $1+\mu\lambda_m\geq 0, \ m=1,\cdots,M$, and $\chi'(\mu)<0$. Therefore, $\chi(\mu)$ is monotonic in the possible region of the solution, and any local solution is guaranteed to be the unique solution. Moreover, the equation $\chi(\mu)=0$ can be efficiently solved by either bisection search method or Newton's method. 


After obtaining $\tilde{\mathbf q}_k$ from problem (\ref{prob:q_qcqp1_equ}), the optimal $\mathbf q_k$ update is given by 
\begin{align}
    \mathbf q_k^{t+1}=\mathbf Q\tilde{\mathbf q}_k.
\end{align}

2) $\bm\theta$ update: The update of $\bm\theta$ can be obtained by solving the following problem. 
\begin{align}
    \bm\theta^{t+1}&=\argmin_{\bm\theta} \sum_{k=1}^K\|\mathbf q_k^{t+1}-\bm\theta+\mathbf u_k^t\|^2\nonumber\\
    &\st \quad|\theta_m|=1, m=1,\cdots,M.
\end{align}
Thus, the optimal $\bm\theta$ is simply the projection of $\frac{1}{K}\sum_{k=1}^K(\mathbf q_k^{t+1}+\mathbf u_k^t)$ onto the unit-modulus constraints, i.e., 
\begin{align}\label{eq:theta_update}
    \bm\theta^{t+1}=e^{j\angle \frac{1}{K}\sum_{k=1}^K(\mathbf q_k^{t+1}+\mathbf u_k^t)}.
\end{align}

3) $\mathbf u_k$ update: The update of $\mathbf u_k$ is standard dual ascent and is given by
\begin{align}\label{eq:u_update}
    \mathbf u_k^{t+1}=\mathbf u_k^t+\mathbf q_k^{t+1}-\bm\theta^{t+1}.
\end{align}


As a result, the optimal phase-shift matrix can be obtained by jointly exploiting ELS and ADMM.



\subsection{Alternating Optimization Framework}
We summarize the proposed alternating optimization algorithm here. Specifically, the AO algorithm is first initialized by $\mathbf w_k^0$ and $\bm\theta^0$. Then, given fixed $\mathbf w_k^t$ and $\bm\theta^t$ in the $t$-th iteration, $\mathbf w_k^{t+1}$ and $\bm\theta^{t+1}$ in the $(t+1)$-th iteration are updated alternatively. Moreover, the convergence of the AO algorithm is demonstrated in Lemma \ref{lem:AO_converge}.

\begin{Lemma}\label{lem:AO_converge}
    With the AO algorithm, the objective value of $\mathcal P1$ is non-increasing in the consecutive iterations.
\end{Lemma}
\begin{proof}
    Please refer to Appendix B.
\end{proof}

\section{Simulation Results}
In this section, we evaluate the performance of our proposed algorithms via simulations. We consider 4 users each with a learning task. The 4 learning tasks considered herein are SVM, CNN with MNIST dataset, CNN with Fashion-MNIST dataset and PointNet. The number of BS antennas varies from 10 to 50, and the number of reflecting elements of the RIS is set to 50. The total transmission time $T=10$ s, bandwidth $B=5$ MHz, and noise power $\sigma^2=-77$ dBm. All the channels involved are assumed to be Rayleigh fading, and the channel coefficients (i.e., the elements in $\mathbf G$, $\mathbf h_{\text{d},k}$, and $\mathbf h_{\text{r},k}$, for all $k$) are normalized with zero mean and unit variance \cite{Guo2020sumrateRIS}. The pathloss exponent of the direct link, i.e., from BS to the users is 4 and the pathloss exponents of BS-RIS link and RIS-user link are set to 2.2.



\subsection{Parameter Fitting for the Learning Tasks}
In this part, the parameters $c_k$'s and $d_k$'s  in the nonlinear learning error models for the $K$ learning tasks are acquired by least mean square (LMS) fitting. Specifically, the SVM classifier is trained on the digits dataset in the Python Scikit-learn ML toolbox. The dataset contains 1797 images of size $8\times 8$ from 10 classes, with 5 bits (representing integers $0\sim16$) for each pixel. Thus, each images needs $D_k=8\times 8\times 5+4=324$ bits. We train the SVM classifier using the first 1000 image samples with sizes $30,50,100,200,300,500,1000$, and use the last 797 image samples for testing. We record the corresponding test errors with different training sample sizes. After that, LMS fitting is applied to obtain $(c_k,d_k)$ for the SVM classifier. Then, we consider a 6-layer CNN with MNIST and Fashion-MNIST datasets, respectively. The CNN consists of a $5\times 5$ convolution layer (with ReLu activation, 32 channels), a $2\times 2$ max pooling layer, another $5 \times 5$ convolution layer (with ReLu activation, 64 channels), a $2\times 2$ max pooling layer, a fully connected layer with 128 units (with ReLu activation), and a final softmax output layer (with 10 outputs). For the MNIST dataset, it consists of 70000 grayscale images (a training set of 60000 examples and a test set of 10000 examples) of handwritten digits, each with $28\times 28$ pixels. Thus, each image needs $D_k=28\times 28\times 8+4=6276$ bits. Each image sample of Fashion-MNIST dataset also needs $D_k=6276$ bits. We train the CNN classifier with sample sizes $100,150,200,300,500,1000,	3000,5000,7000,10000$ for both MNIST and Fashion-MNIST datasets, and record the test errors corresponding to the different training sample sizes. Then, similar LMS fitting is exploited to obtain $(c_k,d_k)$ for these two learning tasks. We also consider PointNet \cite{PointNet} as another learning task to classify 3D point clouds dataset ModelNet40 that contains 12311 CAD models from 40 object categories and  splits into 9843 for training and 2468 for testing. Each data sample has 2000 points with three single-precision floating-point coordinates (4 Bytes). Thus, the data size per sample is $D_k=(2000\times3\times 4+1)\times 8 = 192008$ bits. Similarly, we train the PointNet with sample sizes $1000,3000,5000,7000,9843$, and fit the result to the nonlinear learning error model to obtain $(c_k,d_k)$ for PointNet. The resultant fitting parameters $(c_k,d_k)$ of the nonlinear learning error model are (7.07,0.81), (10.79,0.73), (0.82,0.23) and (0.96,0.24) for SVM, MNIST, Fashion-MNIST and PointNet, respectively.

\begin{figure*}[t]
    \begin{multicols}{3}
        \includegraphics[width=0.95\linewidth]{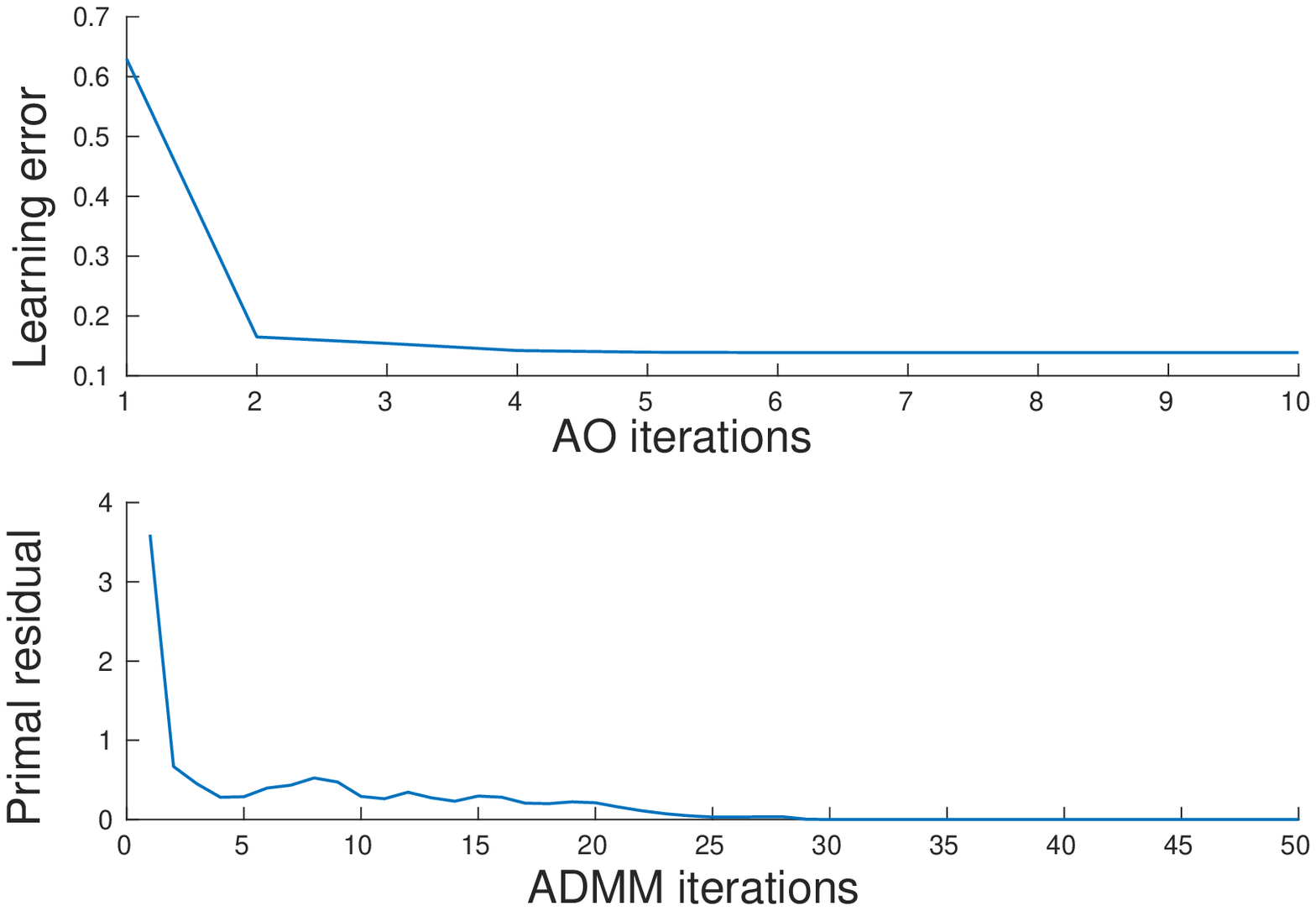}\par\caption{Convergence of the AO and ADMM algorithms.}
        \label{fig:converge}
        \includegraphics[width=0.95\linewidth]{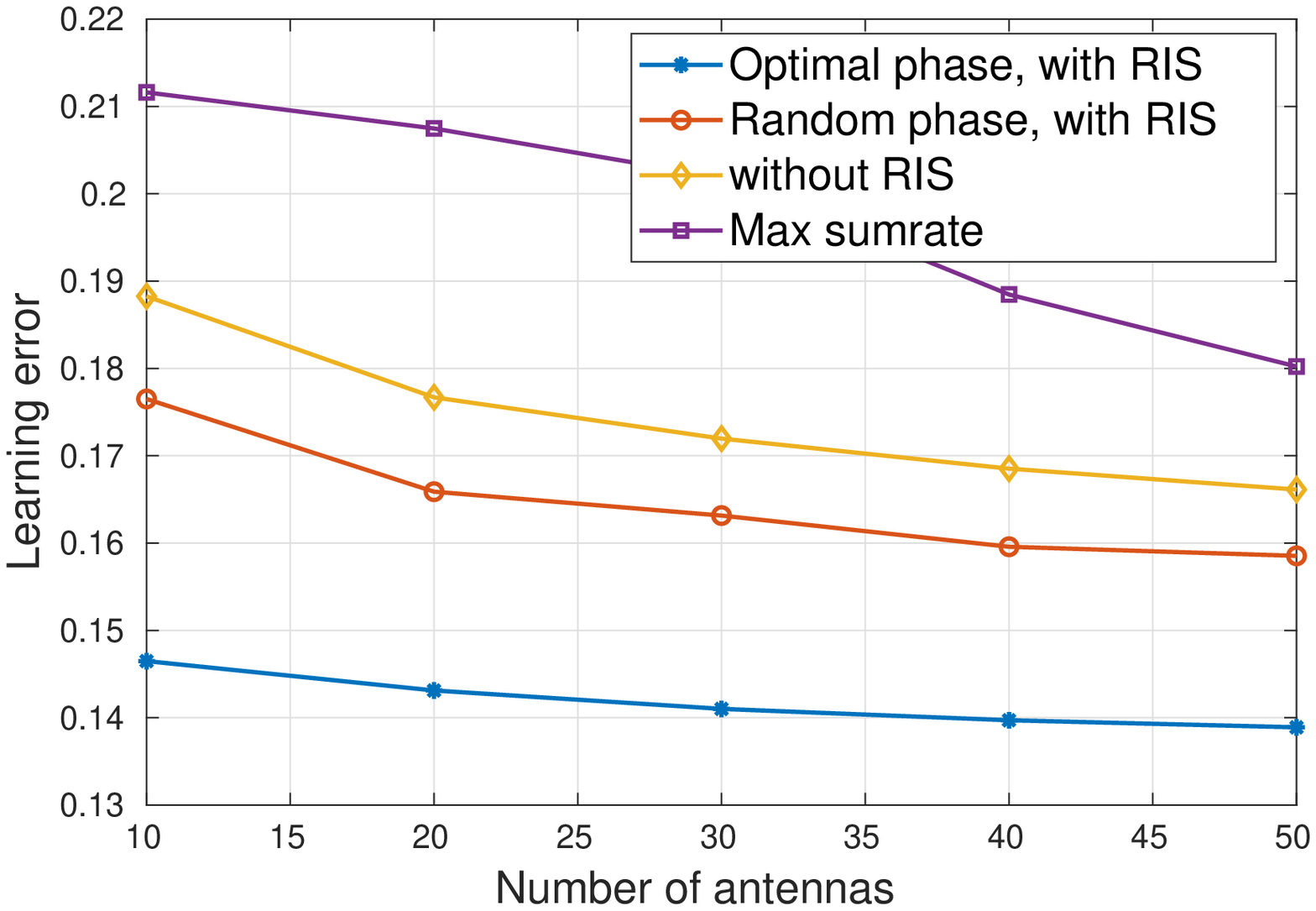}\par\caption{Learning error comparison of various benchmarks.}\label{fig:benchmarks}
        \includegraphics[width=0.95\linewidth]{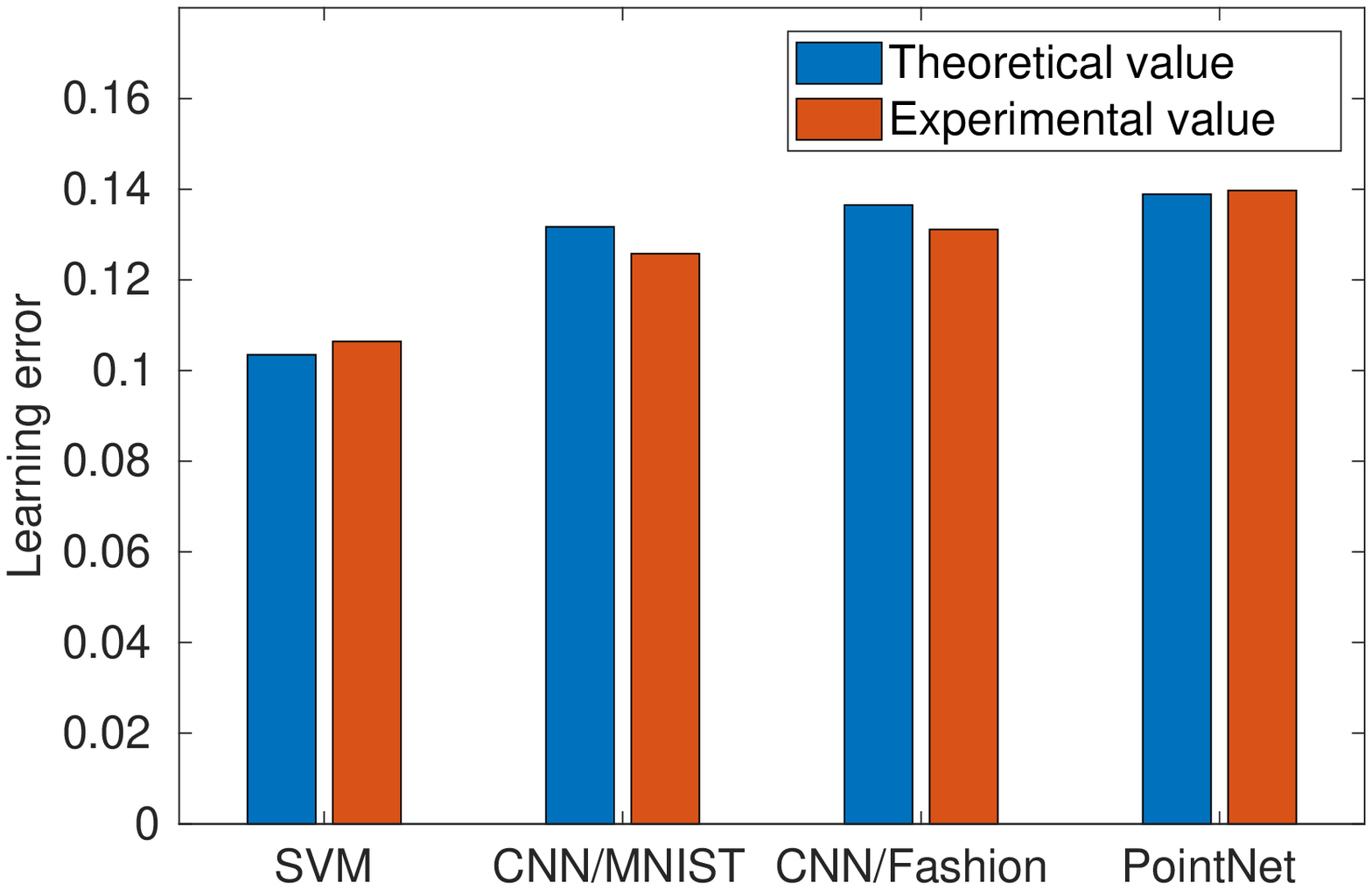}\par\caption{Theoretical learning errors v.s. Experimental learning errors.}\label{fig:experiment}
    \end{multicols}
    \end{figure*}

\subsection{Convergence of AO and ADMM algorithms}
The convergence of the AO algorithm has been proved theoretically and we further show it by simulations here. The top of Fig. \ref{fig:converge} shows that the value of the objective function is non-increasing in the consecutive AO iterations, and converges after around 4 iterations, which is quite efficient.  Moreover, the convergence of the ADMM algorithm is also verified by simulations. It is shown in the bottom of Fig. \ref{fig:converge} that the primal residual concussively degrades and the ADMM algorithm converges after around 30 iterations.

\subsection{Comparison with Various Benchmarks}
We demonstrate the superiority of our RIS-assisted learning-centric scheme with various benchmarks in Fig. \ref{fig:benchmarks}. The three benchmarks considered in this paper are: 1) without deploying the RIS, 2) deploying the RIS with random phase-shift matrix, and 3) maximizing the sumrate as in conventional communication systems. It is shown that the performances of learning-centric schemes are always dramatically better than that of conventional sumrate-maximization scheme, even without the help of the RIS, which demonstrates the necessity of redesign of the wireless communication systems in learning-driven scenarios. Also shown in Fig. \ref{fig:benchmarks} is that with the presence of the RIS, the learning performance can be improved remarkably, justifying the gain of deploying the RIS. Moreover, it can be seen that our proposed phase-shift optimization can further improve the learning accuracy significantly, validating the effectiveness of our proposed optimization algorithms. 

To demonstrate the validity of the nonlinear learning error model, we compare the learning errors obtained from the theoretical error model with those obtained from real experiments. Specifically, we record the optimal number of data samples for each ML task and the corresponding theoretical learning error. Then, we use the optimized sample sizes to train the corresponding learning models, and average the resulting learning errors from 10 runs to obtain the experimental learning errors. Fig. \ref{fig:experiment} shows that the theoretical results conform to the experimental results very well.

\section{Conclusions}
We have investigated the RIS-assisted mobile edge computing systems with learning tasks. The design of a learning-efficient system was achieved by jointly optimizing  the beamforming vectors of the BS and the phase-shift matrix of the RIS in an AO framework. Efficient algorithms were elaborated to address the highly nonconvex optimization problem induced by the nonlinear learning error model and unit-modulus constraints of RIS elements. Experimental results demonstrated the validity of the learning error model and superiority of our proposed scheme over various benchmarks. 

\appendix

\subsection{Proof of Lemma \ref{lem:opt_q}}
Since strong duality holds for QCQP problems with one constraint as proved in \cite{boyd2004convex}, we can solve the dual problem of (\ref{prob:q_qcqp1_equ}). The Lagrangian of (\ref{prob:q_qcqp1_equ}) is 
\begin{align*}
    \mathcal L(\mathbf q,\mu)=\|\tilde{\mathbf q}-\bm\zeta^t\|^2+\mu(\tilde{\mathbf q}\hermconj\mathbf\Lambda\tilde{\mathbf q}-2\re\{\tilde(\mathbf b)\hermconj\tilde{\mathbf q}\}-\tau).
\end{align*}
Setting $\frac{\partial \mathcal L(\mathbf q,\mu)}{\partial \mathbf q}=0$, we obtain the optimal $\tilde{\mathbf q}$ as 
\begin{align*}
    \tilde{\mathbf q}^*=(\mathbf I+\mu\mathbf\Lambda)^{-1}(\tilde{\bm\zeta}+\mu\tilde{\mathbf b}).
\end{align*}
Substituting the above equation back to the equality constraint in (\ref{prob:q_qcqp1_equ}), it becomes a nonlinear equation with respect to $\mu$:
\begin{align*}
    \chi(\mu)\!=\!\sum_{m=1}^M\lambda_m\left|\frac{\tilde{\zeta}_m+\mu\tilde b_m}{1+\mu\lambda_m}\right|^2\!\!-\!2\re\left\{\sum_{m=1}^M\tilde b_m^*\frac{\tilde{\zeta}_m+\mu\tilde b_m}{1+\mu\lambda_m}\right\}\!-\!\tau,
\end{align*}
where $\lambda_m$ is the $m$-th diagonal entry of $\mathbf\Lambda$. 
\subsection{Proof of Lemma \ref{lem:AO_converge}}
For ease of notation, we denote the objective function of $\mathcal P1$ as $g(\mathbf w,\bm\theta)$. Assume $\mathbf w^t$ and $\bm\theta^t$ are obtained by the corresponding optimization problems in the $t$-th iteration, respectively. Then, we have 
\begin{align*}
g(\mathbf w^t,\bm\theta^{t+1})=\min_{\bm\theta} g(\mathbf w^t,\bm\theta)\leq g(\mathbf w^t,\bm\theta^t).
\end{align*}
Analogously, it holds that
\begin{align*}
    g(\mathbf w^{t+1},\bm\theta^{t+1})=\min_{\mathbf w} g(\mathbf w,\bm\theta^{t+1})
    \leq g(\mathbf w^t,\bm\theta^{t+1})
    \leq g(\mathbf w^t,\bm\theta^t).
\end{align*}

\bibliographystyle{IEEEtran}
\bibliography{Reference_IRS.bib}
\end{document}